\pdfoutput=1
\RequirePackage{ifpdf}
\ifpdf 
\documentclass[pdftex]{sigma}
\else
\documentclass{sigma}
\fi

\numberwithin{equation}{section}

\newtheorem{Theorem}{Theorem}[section]
\newtheorem{Corollary}[Theorem]{Corollary}

\newtheorem{Proposition}[Theorem]{Proposition}
 { \theoremstyle{definition}

 }

\usepackage{epic}

\def\t#1{\widetilde{#1}}
\def\h#1{\widehat{#1}}
\def\hh#1{\widehat{\widehat{#1}}}
\def\th#1{\widehat{\widetilde{#1}}}
\def\b#1{\overline{#1}}

\def\uh#1{\underset{\text{$\widehat{\phantom{\cdot}}$}}{#1}}

\newcommand{\br}{\boldsymbol{r}}
\newcommand{\bc}{\boldsymbol{c}}

\newcommand{\bI}{\boldsymbol{I}}
\newcommand{\bM}{\boldsymbol{M}}
\newcommand{\bK}{\boldsymbol{K}}
\newcommand{\ba}{\boldsymbol{a}}
\newcommand{\bb}{\boldsymbol{b}}
\newcommand{\bd}{\boldsymbol{d}}
\newcommand{\bu}{\boldsymbol{u}}

\usepackage{xcolor}

\begin{document}
\allowdisplaybreaks

\newcommand{\arXivNumber}{2106.12835}

\renewcommand{\PaperNumber}{093}

\FirstPageHeading

\ShortArticleName{A Revisit to the ABS H2 Equation}

\ArticleName{A Revisit to the ABS H2 Equation}

\Author{Aye Aye CHO, Maebel MESFUN and Da-Jun ZHANG}

\AuthorNameForHeading{A.A.~Cho, M.~Mesfun and D.-J.~Zhang}

\Address{Department of Mathematics, Shanghai University, Shanghai 200444, P.R.~China}
\Email{\href{mailto:ayeaye1@shu.edu.cn}{ayeaye1@shu.edu.cn},
\href{mailto:maebel@126.com}{maebel@126.com},
\href{mailto:djzhang@staff.shu.edu.cn}{djzhang@staff.shu.edu.cn}}

\ArticleDates{Received June 25, 2021, in final form October 13, 2021; Published online October 18, 2021}

\Abstract{In this paper we revisit the Adler--Bobenko--Suris H2 equation. The H2 equation is linearly related to the $S^{(0,0)}$ and $S^{(1,0)}$ variables in the Cauchy matrix scheme. We elaborate the coupled quad-system of $S^{(0,0)}$ and $S^{(1,0)}$ in terms of their 3-dimensional consistency, Lax pair, bilinear form and continuum limits. It is shown that $S^{(1,0)}$ itself satisfies a 9-point lattice equation and in continuum limit $S^{(1,0)}$ is related to the eigenfunction in the Lax pair of the Korteweg--de Vries equation.}


\Classification{35Q51; 35Q55; 37K60}

\section{Introduction}

In the past two decades the research of discrete integrable systems has undergone a true development
(e.g., see \cite{HJN-2016} and the references therein).
The property called multi-dimensional consistency (MDC) plays a central role
in understanding integrability of discrete systems \cite{ABS-2003,BS-2002,Nij-2002,NW-2001}.
The~MDC means a lattice equation can be consistently embedded into a higher dimensional system.
For quadrilateral equations, this property is geometrically described as the
consistency around a~cube (CAC).
Adler, Bobenko and Suris (ABS) classified all affine linear quadrilateral equations
that are consistent-around-cube (CAC), $D_4$ symmetric and possess the so-called tetrahedron property~\cite{ABS-2003}.
In their full list there are amazingly 9 equations: namely,
H1, H2, H3$_\delta$, A1$_\delta$, A2, Q1$_\delta$, Q2, Q3$_\delta$ and Q4.
Most of the equations in the ABS list are (or related to) equations known before the list was found.
One of the new equation is H2, which is also special because only for~H2 and H3$_\delta$
their continuous counterpart are not found in terms of Miwa's coordinates~\cite{Vermeeren-SIGMA-2019}.

H2 equation reads \cite{ABS-2003}
\begin{gather}\label{H2}
\big(y-\th y\big)\big(\t{y}-\h{y}\big)+\big(p^{2}-q^{2}\big)\big(y+\t{y}+\h{y}+\th{y}-p^{2}-q^{2}\big)=0,
\end{gather}
where $p$, $q$ are spacing parameters in $n$- and $m$-direction, respectively, and
tilde-hat notations stand for shifts in $n$- and $m$-direction (see equation~\eqref{shift} for definition).
In this paper, we will revisit H2 equation with the help of H1
(i.e., the lattice potential Korteweg--de Vries equation (lpKdV)~\cite{NQC-1983})
\begin{gather}\label{H1}
\big(z-\th z\big)\big(\h{z}-\t z\big)=p^{2}-q^{2}.
\end{gather}
Note that H1 has a background solution $z_0=pn+qm$. By removing the background from the equation,
i.e., introducing $z=z_0-u$ so that $u$ has a zero background, the resulting equation~\cite{NQC-1983}
\begin{gather}\label{lpKdV}
 \big(p+q+u-\th u \big)\big(p-q+\h u-\t u\big)=p^2-q^2
\end{gather}
is ready to yield the potential KdV equation in its continuum limit.
H2 has also a background solution $y_0=z_0^2$ \cite{HZ-2009} but it is hard to remove it for H2 equation.
This is why so far the continuous counterpart of H2 is unknown.

The main contents of the paper are the following.
We will start from the following relations obtained from
Cauchy matrix approach (see equation~\eqref{BT-1-eq}),
\begin{subequations}\label{BT-1}
\begin{gather}
w+\widetilde{w} = p\widetilde{u}-pu+u\widetilde{u},
\\
w+\widehat{w} = q\widehat{u}-qu+u\widehat{u},
\end{gather}
\end{subequations}
which provide a non-auto B\"acklund transformation (BT) between H1 and H2 with
\begin{subequations}\label{yzuw}
\begin{gather}
y=z_0^2-2z_0 u+ 2w, \label{yzuw-y}
\\
z=z_0-u.
\end{gather}
\end{subequations}
From \eqref{BT-1} we will derive a quadrilateral coupled system
(in Section~\ref{sec-3-1})
\begin{subequations}\label{uw-eq}
\begin{gather}
\big(p+q+u-\widehat{\widetilde{u}}\big)\big(\widetilde{w}-\widehat{w}\big)
-(p-q)\big(w-\widehat{\widetilde{w}}\big)=0, \label{uw-eq-a}
\\
\big(p-q+\widehat{u}-\widetilde{u}\big)\big(w-\widehat{\widetilde{w}}\big)
-(p+q)\big(\widetilde{w}-\widehat{w}\big)=0, \label{uw-eq-b}
\end{gather}
\end{subequations}
which is CAC and therefore is integrable.
As a consequence we will have a one-component 9-point equation of $w$,
\begin{gather}\label{w-eq}
(p+q)\left(\frac{\t{w}-\h{w}}{w-\th{w}}+\frac{\h{\h{\t{w}}}-
\h{\t{\t{w}}}}{\h{\t{w}}-\h{\h{\t{\t{w}}}}}\right)
+(p-q)\left(\frac{\h{w}-\h{\h{\t{w}}}}{\h{\h{w}}-\h{\t{w}}}+
\frac{\t{w}-\h{\t{\t{w}}}}{\t{\t{w}}-\h{\t{w}}}\right)=0.
\end{gather}
Considering the relation \eqref{yzuw}
we call \eqref{uw-eq} or \eqref{w-eq} an alternative of H2.
In this paper we will focus on the coupled system \eqref{uw-eq} to investigate its
integrability, bilinear form, continuum limits, etc.
It turns out that the continuous counterpart of \eqref{uw-eq} consists of
the potential KdV equation and the time-part in the Lax pair of the (potential) KdV equation.
Note that with some transformations the system \eqref{uw-eq}
will be equivalent to a $\mathrm{H}1\times \mathrm{H}2$ two-component system
which is first derived in \cite{KNPT-2020} (see equation~\eqref{H1H2}).

The paper is arranged as follows.
In Section~\ref{sec-2} we recall the Cauchy matrix approach and derive the BT \eqref{BT-1}.
Then we derive the coupled quad-system \eqref{uw-eq} and investigate its 3D consistency and Lax pairs in Section~\ref{sec-3},
and present its bilinear form and Casoratian solutions in Section~\ref{sec-4}.
After that, Section~\ref{sec-5} consists of the investigation of the continuous counterpart of the quad-system \eqref{uw-eq},
including their continuum limits and Cauchy matrix approach in continuous case,
which leads to a connection between $S^{(1,0)}$ equation and the KdV equation.
Finally, concluding remarks are given in Section~\ref{sec-6}.
There is an appendix section to explore higher order equations in the continuum limits.

\section{Cauchy matrix scheme}\label{sec-2}

\subsection{Preliminary}

For a function $U$ of $(n,m)\in \mathbb{Z}^2$, we introduce shorthand to express it and its shifts, such as
\begin{gather}\label{shift}
U\doteq U(n,m),\quad\
\t U\doteq U(n+1,m),\quad\
\h U\doteq U(n,m+1),\quad\
\th U\doteq U(n+1,m+1).
\end{gather}
Spacing parameters $p$ and $q$ serve as spacing parameters of $n$-direction and $m$-direction, respectively.
These settings are depicted as in Figure~\ref{Fig-1}($a$).
In 3-dimensional case, we introduce the third direction $l$ together with its spacing parameter $r$,
and denote the shift by $\b U \doteq U(n,m,l+1)$. We also introduce shift operators $E_n$, $E_m$ and $E_l$ such that
$E_n U= \t U$, $E^2_mU=\h{\h U}$, etc.

\begin{figure}[ht]
\setlength{\unitlength}{0.0004in}
\hspace{2cm}
\begin{picture}(3482,2813)(0,-10)
\put(1925,-510){\makebox(0,0)[lb]{($a$)}}
\put(1275,2708){\circle*{150}}
\put(825,2808){\makebox(0,0)[lb]{$\h U$}}
\put(3075,2708){\circle*{150}}
\put(3275,2808){\makebox(0,0)[lb]{$\th U$}}
\put(1275,908){\circle*{150}}
\put(945,1808){\makebox(0,0)[lb]{$q$}}
\put(825,1008){\makebox(0,0)[lb]{$U$}}
\put(3075,908){\circle*{150}}
\put(3300,1008){\makebox(0,0)[lb]{$ \t U$}}
\put(2025,1008){\makebox(0,0)[lb]{$p$}}
\drawline(275,2708)(4075,2708)
\drawline(3075,3633)(3075,0)
\drawline(275,908)(4075,908)
\drawline(1275,3633)(1275,0)
\end{picture}
\hspace{3cm}
\begin{picture}(3482,3700)(0,-500)
\put(1325,-1000){\makebox(0,0)[lb]{($b$)}}
\put(450,1883){\circle*{150}}
\put(-100,1883){\makebox(0,0)[lb]{$\widetilde{\overline{U}}$}}
\put(1275,2708){\circle*{150}}
\put(825,2708){\makebox(0,0)[lb]{$\overline{U}$}}
\put(3075,2708){\circle*{150}}
\put(3375,2633){\makebox(0,0)[lb]{$\h{\b{U}}$}}
\put(2250,83){\circle*{150}}
\put(2650,8){\makebox(0,0)[lb]{$\widehat{\widetilde{U}}$}}
\put(1275,908){\circle*{150}}
\put(1275,908){\circle*{150}}
\put(825,908){\makebox(0,0)[lb]{$U$}}
\put(1025,1508){\makebox(0,0)[lb]{$r$}}
\put(2000,633){\makebox(0,0)[lb]{$q$}}
\put(560,400){\makebox(0,0)[lb]{$p$}}
\put(2250,1883){\circle*{150}}
\put(1850,2008){\makebox(0,0)[lb]{$\widehat{\widetilde{\overline{U}}}$}}
\put(450,83){\circle*{150}}
\put(0,8){\makebox(0,0)[lb]{$ \widetilde{U}$}}
\put(3075,908){\circle*{150}}
\put(3300,833){\makebox(0,0)[lb]{$\widehat{U}$}}
\drawline(1275,2708)(3075,2708)
\drawline(1275,2708)(450,1883)
\drawline(450,1883)(450,83)
\drawline(3075,2708)(2250,1883)
\drawline(450,1883)(2250,1883)
\drawline(3075,2633)(3075,908)
\dashline{60.000}(1275,908)(450,83)
\dashline{60.000}(1275,908)(3075,908)
\drawline(2250,1883)(2250,83)
\drawline(450,83)(2250,83)
\drawline(3075,908)(2250,83)
\dashline{60.000}(1275,2633)(1275,908)
\end{picture}\vspace{2ex}
\caption{($a$): The points on which the equation is defined and ($b$): the
 consistency cube.}\label{Fig-1}
\end{figure}
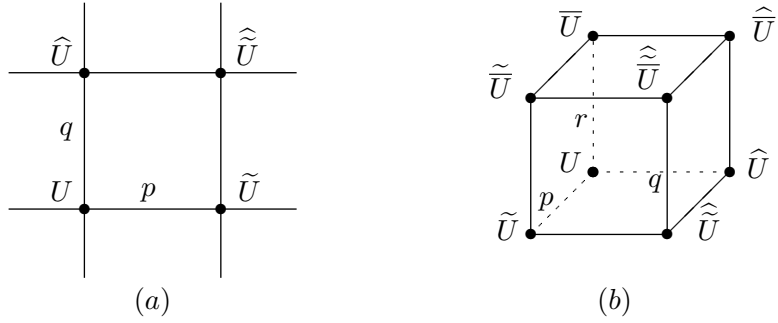

A quadrilateral equation
\[
Q\big(U,\t U, \h U, \th U; p,q\big)=0
\]
is CAC means the equation can be consistently embedded on six faces of the cube in Figure~\ref{Fig-1}($b$).
In other words, given initial values $U$, $\t U$, $\h U$, $\b U$,
the value of $\th{\b U}$ is unique although it will be determined by the top-, front- and right-side, three equations.

\subsection[Cauchy matrix approach to~(1.4)]
{Cauchy matrix approach to (\ref{BT-1})}

To derive the coupled system \eqref{BT-1}, we consider the Sylvester equation
\begin{gather}\label{Sylvester}
\bK \bM + \bM\bK = \br {\bc}^{\rm T},
\end{gather}
where $\bM, \bK \in \mathbb{C}_{N\times N}$, $\br$ and $\bc$ are column vectors in $\mathbb{C}_N$ defined as
\[
\br=(r_1, r_2,\dots, r_N^{})^{\rm T},\qquad \bc=(c_1, c_2,\dots, c_N^{})^{\rm T}.
\]
Let us sketch some known results related to this equation.

\begin{Proposition}
For given $\br$, $\bc$ and $\bK$, if $\bK$ and $-\bK$ do not share eigenvalues, then $\bM$ will be uniquely
determined by the equation \eqref{Sylvester} {\rm \cite{Syl-1884}}.
In addition, explicit formulae of $\bM$ have been constructed in {\rm \cite{ZZ-2013}}
with respect to canonical forms of $\bK$.
\end{Proposition}

\begin{Proposition}[\cite{ZZ-2013}]\label{prop-2}
Introduce
\begin{gather}
S^{(i,j)} = \bc^{\rm T}_{} {\bK}^{j}(\bI+\bM)^{-1} \bK^i \br,\qquad i,j\in \mathbb{Z},
\label{S-ij}
\end{gather}
where the elements $\br$, $\bc$, $\bK$ and $\bM$ satisfy the Sylvester equation \eqref{Sylvester}
and $\bI$ is the $N$-th order unit matrix.
For given $\br$ and $\bc$, the matrix $\bK$ and any matrix similar to it
give rise to the same~$S^{(i,j)}$, and the symmetric property
\begin{gather}\label{S-sym}
S^{(i,j)}=S^{(j,i)}
\end{gather}
holds.
\end{Proposition}

The Cauchy matrix approach is a method to construct and study integrable equations
by means of the Sylvester-type equations.
In this approach integrable equations are presented as closed forms of $S^{(i,j)}$ and its shifts (or derivatives).
It is first systematically used in \cite{NAH-2009} to investigate ABS equations and later
developed in \cite{XZZ-2014,ZZ-2013} to more general cases.

To derive coupled system \eqref{BT-1}, we impose dispersion relation on $\br$ as
\begin{gather*}
 (p\bI-\bK) \widetilde{\br}=(p\bI+\bK)\br,\qquad (q\bI-\bK) \h{\br}=(q\bI+\bK)\br,
\end{gather*}
and let $\bc$ be a constant column vector, where $\bK$ is a constant $N\times N$ matrix.
Next, making use of the results obtained in \cite{ZZ-2013}, we have the following shift relations for $\bM$,
\begin{gather*}
\t{\bM}(p\bI+\bK)-(p\bI+\bK)\bM =\t \br \bc^{\rm T},\\
(p\bI-\bK)\t \bM -\bM (p\bI-\bK) =\br \bc^{\rm T},\\
\h{\bM}(q\bI+\bK)-(q\bI+\bK)\bM =\h \br \bc^{\rm T},\\
(q\bI-\bK)\h \bM -\bM (q\bI-\bK) =\br \bc^{\rm T},
\end{gather*}
from which one arrives at (see equations~(27) in \cite{ZZ-2013})
\begin{gather*}
p\widetilde{S}^{(i,j)}-\widetilde{S}^{(i,j+1)} = pS^{(i,j)}+S^{(i+1,j)}-\widetilde{S}^{(i,0)}S^{(0,j)}, 
\\
q\widehat{S}^{(i,j)}-\widehat{S}^{(i,j+1)} = qS^{(i,j)}+S^{(i+1,j)}-\widehat{S}^{(i,0)}S^{(0,j)}. \label{S-c}
\end{gather*}
Now, taking $(i,j)=(0,0)$ in the above recurrence relation, defining
\begin{gather}\label{uw-S}
u=S^{(0,0)},\qquad w=S^{(1,0)},
\end{gather}
and making use of the symmetric property \eqref{S-sym},
we have
\begin{subequations}\label{BT-1-eq}
\begin{gather}
w+\widetilde{w} = p\widetilde{u}-pu+u\widetilde{u},\label{BT-1a}\\
w+\widehat{w} = q\widehat{u}-qu+u\widehat{u}, \label{BT-1b}
\end{gather}
\end{subequations}
i.e., the coupled system \eqref{BT-1}.

\subsection{Consistent triplet}

Suppose that we have a coupled system
\begin{gather}\label{BT-f}
f\big(u,v,\t u, \t v,p\big)=0,\qquad
f\big(u,v,\h u, \h v,q\big)=0.
\end{gather}
The compatibility with respect to $v$ \big(i.e., $\th v=\t{\h v}$\big) gives rise to a quadrilateral equation of $u$,
\begin{gather}
Q\big(u,\t u, \h u, \th u; p,q\big)=0,\label{u}
\end{gather}
while the compatibility for $u$ yields a quadrilateral equation of $v$,
\begin{gather}
G\big(v,\t v, \h v, \th v; p,q\big)=0.\label{U}
\end{gather}
When this happens we say that equations~\eqref{BT-f}, \eqref{u} and \eqref{U} compose a consistent triplet.
Note that such a triplet can be extended to the lattice systems defined on larger stencils,
not necessary to be restricted to the quadrilateral case.
In such a triplet, equation~\eqref{BT-f} acts as a~BT to connect the other two equations \eqref{u} and \eqref{U},
and due to the compatibility, any pair of solutions $(u,v)$ of \eqref{BT-f} will provide solutions
to \eqref{u} and \eqref{U}.
Consistent triplets have played a useful role in constructing rational solutions for a number of
quadrilateral equations \cite{ZhaZ-SIGMA-2017}.

Using the transformation \eqref{yzuw}, the coupled system \eqref{BT-1-eq} is written as
\begin{subequations}\label{BT-2}
\begin{align}
& y+\widetilde{y}-p^2 = 2 z \widetilde{z}, \label{BT-2a}\\
& y+\widehat{y}-q^2 = 2z\widehat{z}. \label{BT-2b}
\end{align}
\end{subequations}
Note that \eqref{yzuw} has been implied from equation~(5.26) in \cite{NAH-2009}.
By checking compatibility of $z$ and $y$, one immediately find that $y$ and $z$ have to satisfy
H2 \eqref{H2} and H1 \eqref{H1}, respectively.
Thus we have the following.
\begin{Proposition}\label{prop-333}
{\rm H2} equation \eqref{H2}, {\rm H1} \eqref{H1} and the system \eqref{BT-2} compose a consistent triplet.
\end{Proposition}
This means that solving \eqref{BT-2} yields solutions to H1 and H2.
Based on this fact in Section~\ref{sec-4} we will bilinearize \eqref{BT-2}
and get a bilinear form for H2 which is simpler than the one given in~\cite{HZ-2009}.
Note that the fact of \eqref{BT-2} being a BT for H1 and H2 was also found in \cite{Atk-JPA-2008}
using the CAC property of H2 and the degeneration relation from H2 to H1.

\section[Coupled system~(1.6)]
{Coupled system (\ref{uw-eq})}\label{sec-3}

The coupled system \eqref{BT-1-eq} is a set of 3-point equations.
Initial values can be given on any line $n+m=l$, where $l$ can be arbitrary integers.
However, the system can only evolve towards down-left from the initial-value line.
Initial values can also be given on a vertical line or on a~horizontal line,
but neither of them can yield a bidirectional evolution.
Compared with \eqref{BT-1-eq}, the quad-system \eqref{uw-eq} is not only CAC,
but also allows bidirectional evolution starting from initial values
given on some staircase.
Next, we derive \eqref{uw-eq} from \eqref{BT-1-eq} and investigate its integrability.

\subsection{Derivation}\label{sec-3-1}

First, in light of the compatibility $E_n E_m w= E_m E_n w$,
from \eqref{BT-1-eq} we immediately get H1 (lpKdV) equation \eqref{lpKdV}.
Next, eliminating $w$ from \eqref{BT-1-eq} by $\eqref{BT-1b}$--$\eqref{BT-1a}$ and by
$\h{\eqref{BT-1a}}$--$\t{\eqref{BT-1b}}$, respectively, we get
\begin{gather}\label{w-th-1}
 \h{w}- \t{w}+p\widetilde{u}-q\widehat{u}=\big(p-q+\widehat{u}-\widetilde{u}\big)u
\end{gather}
and
\begin{gather}\label{w-th-2}
 \h{w}-\t{w}+p\widehat{u}-q\widetilde{u}=\big(p-q+\widehat{u}-\widetilde{u}\big)\widehat{\widetilde{u}}.
\end{gather}
Meanwhile, subtracting (\ref{BT-1b}) from $\h{\eqref{BT-1a}}$ and subtracting (\ref{BT-1a})
from $\t{\eqref{BT-1b}}$, respectively, yield
\begin{gather}\label{w-wth-1}
 w-\widehat{\widetilde{w}} + p\widehat{\widetilde{u}}+qu=\big(p+q+u-\widehat{\widetilde{u}}\big)\widehat{u}
\end{gather}
and
\begin{gather}\label{w-wth-2}
 w-\widehat{\widetilde{w}}+q\widehat{\widetilde{u}}+pu=\big(p+q+u-\widehat{\widetilde{u}}\big)\widetilde{u}.
\end{gather}
Then, by multiplying $\th u$ to \eqref{w-th-1} and $u$ to \eqref{w-th-2} and eliminating their
right-hand sides, we arrive~at
\begin{gather*}
 \big(\h{w}-\t{w}\big)\big(u-\widehat{\widetilde{u}}\big)
 =\big(p\t{u}-q\h{u}\big)\widehat{\widetilde{u}}+\big(q\widetilde{u}-p\widehat{u}\big)u,
\end{gather*}
and similarly, from \eqref{w-wth-1} and \eqref{w-wth-2} we have
\begin{gather*}
 \big(\h{u}-\t{u}\big)\big(w-\th{w}\big)=\big(p\widetilde{u}-q\widehat{u}\big)\widehat{\widetilde{u}}
 +\big(q\widetilde{u}-p\widehat{u}\big)u.
\end{gather*}
It then directly follows that
\begin{gather}\label{wu}
\big(w-\th{w}\big)\big(\t{u}-\h{u}\big)=\big(\widetilde{w}-\widehat{w}\big)\big(u-\widehat{\widetilde{u}}\big).
\end{gather}
Finally, combining the above equation and the lpKdV equation \eqref{lpKdV},
eliminating the term $\t{u}-\h{u}$ yields \eqref{uw-eq-a} and
eliminating the term ${u}-\th{u}$ gives rise to \eqref{uw-eq-b}.

Now we have obtained the coupled system \eqref{uw-eq}, together with H1 \eqref{lpKdV} and a relation \eqref{wu}.
The above derivation indicates the following more precise connection of them.

\begin{Proposition}\label{prop-3}
Among the following four equations,
\begin{subequations}\label{uw-sys}
\begin{gather}
\big(p+q+u-\widehat{\widetilde{u}}\big)\big(\widetilde{w}-\widehat{w}\big)
-(p-q)\big(w-\widehat{\widetilde{w}}\big)=0, \label{uw-eq-aa}
\\
\big(p-q+\widehat{u}-\widetilde{u}\big)\big(w-\th{w}\big)-(p+q)\big(\widetilde{w}-\widehat{w}\big)=0, \label{uw-eq-bb}
\\
\big(w-\th{w}\big)\big(\t{u}-\h{u}\big)=\big(\widetilde{w}-\widehat{w}\big)\big(u-\widehat{\widetilde{u}}\big), \label{uw-eq-ccc}
\\
\big(p+q+u-\th u \big)\big(p-q+\h u-\t u\big)=p^2-q^2,\label{uw-eq-dd}
\end{gather}
\end{subequations}
any two of them can be derived from the other two.
\end{Proposition}

\begin{proof}
Let us sketch the proof.
First, eliminating $w$ from \eqref{uw-eq-aa} and \eqref{uw-eq-bb} yields the lpKdV equation
\eqref{uw-eq-dd}. Rewriting \eqref{uw-eq-dd} as
\begin{gather}\label{x1}
 p+q+u-\th u =-(p-q)\frac{u-\th u}{\h u-\t u}
\end{gather}
and then substituting its left-hand-side into \eqref{uw-eq-aa} give rise to \eqref{uw-eq-ccc}.
Thus, \eqref{uw-eq-ccc} and \eqref{uw-eq-dd} are derived if \eqref{uw-eq-aa} and \eqref{uw-eq-bb} hold.

Next, assume \eqref{uw-eq-aa} and \eqref{uw-eq-ccc} are known.
Eliminating $\t w-\h w$ from them yields \eqref{x1}, which is \eqref{uw-eq-dd}.
Note that \eqref{uw-eq-dd} can also be written as
\begin{gather}\label{x2}
 u-\th u = \frac{- (p+q)\big(\h u-\t u\big) }{p-q+\h u-\t u}.
\end{gather}
Substituting its left-hand-side into \eqref{uw-eq-ccc} we get \eqref{uw-eq-bb}.

If we assume \eqref{uw-eq-aa} and \eqref{uw-eq-dd} are known, we may eliminate
$p+q+u-\th u$ from them and get equation \eqref{uw-eq-bb}. Then, as we described before,
\eqref{uw-eq-ccc} can be obtained when \eqref{uw-eq-aa} and \eqref{uw-eq-bb} hold.

If we start from \eqref{uw-eq-ccc} and \eqref{uw-eq-dd},
we can first rewrite \eqref{uw-eq-dd} as
\begin{gather*}
 \h u-\t u =\frac{-(p-q)\big(u-\th u\big)}{p+q+u-\th u},
\end{gather*}
and then eliminate $\h u-\t u$ using \eqref{uw-eq-ccc}. This gives rise to \eqref{uw-eq-aa},
while \eqref{uw-eq-bb} is a consequence of eliminating $u-\th u$ from \eqref{x2} and \eqref{uw-eq-ccc}.

Other cases can be proved similarly.
\end{proof}

In addition, eliminating $u$ from \eqref{uw-eq} yields the nine-point single-component equation \eqref{w-eq}.
Thus we have the following result parallel to Proposition \ref{prop-333}.
\begin{Proposition}\label{prop-33}
{\rm H1} equation \eqref{lpKdV}, the nine-point equation \eqref{w-eq} and the system \eqref{BT-1}
compose a consistent triplet.
\end{Proposition}

Besides, we also note that starting from the coupled system \eqref{uw-eq}
and making use of \eqref{yzuw} and \eqref{BT-2},
one can get another coupled system in terms of $z$ and $y$,
\begin{subequations}\label{H1H2}
\begin{gather}
\big(\t y-\h y\big)\big(\th z-z\big)+\big(p^2-q^2\big)\big(\th z+ z\big)=0,\label{H1H2-a}\\
\big(y-\th y\big)\big(\t z-\h z\big)+\big(p^2-q^2\big)\big(\t z+\h z\big)=0.\label{H1H2-b}
\end{gather}
\end{subequations}
In fact, from \eqref{yzuw} we have
\[ u=z_0-z,\qquad w=\frac{1}{2}\big(y+z_0^2-2z_0z\big).\]
Substituting these into \eqref{uw-eq-a}, and then successfully
using \eqref{BT-2a} to eliminate $y$,
using $\t{\eqref{BT-2b}}$ to eliminate $\th y$
and using \eqref{H1} to eliminate $\h z$,
finally we arrive at \eqref{H1H2-a}.
Equation \eqref{H1H2-b} can be derived in a similar way. This coupled system was first
derived in \cite{KNPT-2020}, known as the $\mathrm{H}1\times \mathrm{H}2$
2-component vertex equation.
Using \eqref{yzuw} and \eqref{BT-2}, the coupled system \eqref{uw-eq} can be recovered
from \eqref{H1H2} as well.

\subsection{Integrability: CAC and Lax pair}

In this part, let us focus on the integrability of the coupled system \eqref{uw-eq}.
In light of Proposition~\ref{prop-3}, in practice, we consider the equivalent equations \eqref{uw-eq-bb} and \eqref{uw-eq-dd},
i.e.,
\begin{subequations}\label{uw-eqeq}
\begin{gather}
\big(p-q+\h{u}-\t{u}\big)\big(w-\th{w}\big)-(p+q)\big(\t{w}-\h{w}\big)=0, \label{uw-eqeq-a}
\\
\big(p+q+u-\th u \big)\big(p-q+\h u-\t u\big)=p^2-q^2. \label{uw-eqeq-b}
\end{gather}
\end{subequations}
Let $U=(u,w)^{\rm T}$ and put the above system on the six faces of the cube (see Figure~\ref{Fig-1}($b$)).
Given initial values $U$, $\t U$, $\h U$, $\b U$, we have
\begin{subequations}\label{u-shift2}
\begin{gather}
\widehat{\widetilde{u}}=u-\frac{(p+q)\big(\t{u}-\h{u}\big)}{p-q+\h{u}-\t{u}}, \label{u-shift2a}
\\
\widetilde{\b{u}}=u-\frac{(p+r)\big(\t{u}-\b{u}\big)}{p-r+\b{u}-\t{u}}, \label{u-shift2b}
\\
\widehat{\b{u}}=u-\frac{(r+q)\big(\b{u}-\h{u}\big)}{r-q+\h{u}-\b{u}}, \label{u-shift2c}
\end{gather}
\end{subequations}
and
\begin{subequations}\label{w-shift2}
\begin{gather}
\widehat{\widetilde{w}}=w-\frac{(p+q)\big(\t{w}-\h{w}\big)}{p-q+\h{u}-\t{u}},\label{w-shift2a}
\\
\widetilde{\b{w}}=w-\frac{(p+r)\big(\t{w}-\b{w}\big)}{p-r+\b{u}-\t{u}},\label{w-shift2b}
\\
\widehat{\b{w}}=w-\frac{(r+q)\big(\b{w}-\h{w}\big)}{r-q+\h{u}-\b{u}}.\label{w-shift2c}
\end{gather}
\end{subequations}
Then, with these in hand, we find that the value of $\th{\b U}$ is same
no matter we calculate it from the equation on the top, or the front, or the right face,
where
\[
\th{\b u}=\frac{B}{A},\qquad \th{\b w}=\frac{C}{A},
\]
with
\begin{align*}
A={}& \big(p^{2}-q^{2}\big)\b{u}+\big(r^{2}-p^{2}\big)\h{u}+\big(q^{2}-r^{2}\big)\t{u}+(q-r)(p-q)(p-r),
\\[.5ex]
B={}& (p-q)(p+r)(q+r)\b{u}+(p+q)(q+r)(r-p)\h{u}+ (p+q)(p+r)(q-r)\t{u}
\\
& {}+\big(r^{2}-q^{2}\big)\b{u}\h{u}+\big(p^{2}-r^{2}\big)\b{u}\t{u}+\big(q^{2}-p^{2}\big)\t{u}\h{u},
 \\[.5ex]
C={}& (p+r)(q+r)(p-q+\h{u}-\t{u})\b{w}+(p+q)(q+r)(r-p+\t{u}-\b{u})\h{w}
\\
& {}+(p+q)(p+r)(q-r+\b{u}-\h{u})\t{w}.
\end{align*}
Thus, the system \eqref{uw-eqeq} (or \eqref{uw-eq}) is CAC.
Note that the lpKdV equation \eqref{uw-eqeq-b} itself is CAC as well.
This also indicates that the 9-point equation \eqref{w-eq} can be embedded into 3-dimensional space
and is consistent around a large cube consisting of 4 elementary cubes.

To obtain a Lax pair, following the standard procedure (cf.~\cite{ABS-2003,BHQK-2013,Nij-2002}),
we introduce
\[
\b{u}=\frac{g}{f},\qquad \b{w}=\frac{h}{f},
\]
and rewrite \eqref{u-shift2b}, \eqref{u-shift2c}, \eqref{w-shift2b} and \eqref{w-shift2c},
which leads us to the matrix system
\begin{gather}\label{Lax-pair}
\widetilde{\Phi}=\mathcal{N}\Phi,\qquad \widehat{\Phi}=\mathcal{M}\Phi,
\end{gather}
where $\Phi=(f,g,h)^{\rm T}$,
\begin{gather*}
\mathcal{N}= \begin{pmatrix}
 p-s-\widetilde{u} & 1 & 0
 \\
 (p-s-\widetilde{u})u-(p+s)\widetilde{u} & p+s+u & 0
 \\
 (p-s-\widetilde{u})w-(p+s)\widetilde{w} & w & p+s
 \end{pmatrix}
\end{gather*}
and
\begin{gather*}
\mathcal{M}= \begin{pmatrix}
 q-s-\widehat{u} & 1 & 0 \\
 (q-s-\widehat{u})u-(q+s)\widehat{u} & q+s+u & 0\\
 (q-s-\widehat{u})w-(q+s)\widehat{w} & w & q+s
 \end{pmatrix}\!.
\end{gather*}
Compatibility of \eqref{Lax-pair}, i.e., $\h{\t \Phi}=\t{\h \Phi}$ yields the four equations in \eqref{uw-sys}.
In light of Proposi\-tion~\ref{prop-3},~\eqref{Lax-pair} provides a Lax pair for any two equations of \eqref{uw-sys},
including \eqref{uw-eq}.

\section{Bilinear formulation}\label{sec-4}

Now that both the coupled system \eqref{uw-eq} and H2 equation \eqref{H2} are the consequence of the
coupled equation \eqref{BT-1}, we can get a bilinear form
for both \eqref{uw-eq} and H2 \eqref{H2} from \eqref{BT-1}.
To~achieve that, we introduce transformation
\begin{gather}\label{trans}
u=\frac{g}{f},\qquad w=\frac{\theta}{2 f},
\end{gather}
from which the coupled system \eqref{BT-1} is bilinearized as
\begin{subequations}\label{BT-1-bil}
\begin{gather}
2p\big(\widetilde{f}g-f\widetilde{g}\big)-2 g\widetilde{g}+\t \theta f+\theta \t{f}=0, \label{BT-1-bil-a}
\\
2q\big(\widehat{f}g-f\widehat{g}\big)-2g\widehat{g}+\h{\theta}f+\theta \h{f}=0. \label{BT-1-bil-b}
\end{gather}
\end{subequations}

In the following we present Casoratian solutions to the above bilinear equations.
We introduce~\cite{HZ-2009}
\begin{gather}\label{psi}
\psi_{i}(n,m,l)=\varrho^{+}_i k_{i}^{l}(p+k_{i})^{n}(q+k_{i})^{m}
+\varrho^{-}_i(-k_{i})^{l}(p-k_{i})^{n}(q-k_{i})^{m},
\end{gather}
where $\varrho^{\pm}_i$ and $k_{i}$ are constants, and define a vector
\begin{gather*}
\psi(n,m,l)=\big(\psi_{1}(n,m,l),\psi_{2}(n,m,l),\dots,\psi_{N}(n,m,l)\big)^{\rm T}.
\end{gather*}
A Casoratian composed by $\psi$ w.r.t shift in $l$ is a determinant defined as (we skip $n$, $m$ notations for convenience)
\begin{align*}
 |\psi(l), \psi(l+1), \psi(l+2),\dots, \psi(l+N-1)|
= |0,1,2,\dots, N-1|=\big|\h{N-1}\big|.
\end{align*}
In a similar way one can introduce
$ |0,1,2,\dots, N-2,N|=\big|\h{N-2},N\big|$, etc.
Note that the shorthand $\big|\h{N-1}\big|$ is a conventional notation in expression of
Casoratian or Wronskian \cite{FreN-PLA-1983}
and we do not think it makes any confusion with the hat-shift for $m$-direction in context.
A useful identity in verifying solution in Caosratian/Wronskian form is \cite{FreN-PLA-1983}
\begin{gather}\label{id-1}
|\bK,\ba,\bb||\bK,\bc,\bd|-|\bK,\ba,\bc||\bK,\bb,\bd|+|\bK,\ba,\bd||\bK,\bb,\bc|=0,
\end{gather}
where $\bK$ is a $N\times (N-2)$ matrix, and $\ba$, $\bb$,
$\bc$ and $\bd$ are $N$th-order column vectors.

Now let us present solutions to the bilinear equation \eqref{BT-1-bil}.

\begin{Theorem}
The equations \eqref{BT-1-bil} provide a bilinear form for {\rm H2} equation \eqref{H2} via the transformation \eqref{yzuw-y}
together with \eqref{trans}.
Equations~\eqref{BT-1-bil} have the following solution
\begin{subequations}\label{fg-theta}
\begin{gather}
f=\big|\widehat{N-1}\big|,\qquad g=\big|\widehat{N-2},N\big|, \label{fg}
\\
\theta =\big|\widehat{N-3},N-1,N\big|+\big|\widehat{N-2},N+1\big|,\label{theta}
\end{gather}
\end{subequations}
composed by the vector $\psi$ with entries \eqref{psi}, in terms of shifts in $l$.
\end{Theorem}

\begin{proof}
Let us first prove \eqref{BT-1-bil-b}. Introduce
\begin{gather}\label{hs}
h=\big|\widehat{N-2},N+1\big|,\qquad s=\big|\widehat{N-3},N-1,N\big|
\end{gather}
and rewrite \eqref{BT-1-bil-b} as
\begin{gather*}
A+B=0,
\end{gather*}
where
\begin{gather*}
A=\h f(h+q g)-\h g(g +q f) +f \h s, 
\qquad
B= f\big(\widehat{h}-q\widehat{g}\big)-g\big(\widehat{g}-q\widehat{f}\big)+\widehat{f}s. 
\end{gather*}
Next, we prove both $A$ and $B$ are zero when $f$, $g$, $h$, $s$ are Casoratians defined above.
Note that with $\psi$ given by \eqref{psi} as elementary entries,
shifts of $f$, $g$, $h$ satisfy (see formulae given in Appendix in \cite{HZ-2009})
\begin{gather*}
q^{N-2}\uh{f}=-\big|\widehat{N-2},\uh{\psi}(N-2)\big|,
\\
q^{N-2}\big(\uh{h}+q\uh{g}\big)=-\big|\widehat{N-3},N,\uh{\psi}(N-2)\big|,
\\
q^{N-2}\big(\uh{g}+q\uh{f}\big)=-\big|\widehat{N-3},N-1,\uh{\psi}(N-2)\big|.
\end{gather*}
Then we have
\begin{align*}
q^{N-2} \uh A={}& q^{N-2}\big[f\big(\uh{h}+q\uh{g}\big)-g\big(\uh{g}+q\uh{f}\big)+\uh{f}s\big]
\\
={}&-\big|\h{N-1}\big|\big|\h{N-3},N,\uh{\psi}(N-2)\big|
+\big|\widehat{N-2},N\big|\big|\widehat{N-3},N-1,\uh{\psi}(N-2)\big|
\\
&-\big|\h{N-2},\uh{\psi}(N-2)\big|\big|\widehat{N-3},N-1,N\big|,
\end{align*}
which vanishes in light of the identity \eqref{id-1},
where $\bK=\big(\widehat{N-3}\big)$, $\ba=\psi(N-2)$, $\bb=\psi(N-1)$, $\bc=\psi(N)$ and $\bd=\uh{\psi}(N-2)|$.
$B=0$ is the equation (4.21e) in \cite{HZ-2009}, which is already proved.
Thus, equation~\eqref{BT-1-bil-b} holds. Equation~\eqref{BT-1-bil-a} can be proved in a similar way.
\end{proof}

The proof of the this theorem indicates that the bilinear equations \eqref{BT-1-bil} can further be decoupled.

\begin{Corollary}
The bilinear form \eqref{BT-1-bil} can be decoupled into
\begin{subequations}\label{BT-1-bil-2}
\begin{gather}
\t f(h+p g)-\t g(g +p f) +f \h s=0,
\\
f\big(\t{h}-p\t{g}\big)-g\big(\t{g}-p\t{f}\big)+\t{f}s=0,
\\
\h f(h+q g)-\h g(g +q f) +f \h s=0,
\\
f\big(\h{h}-q\h{g}\big)-g\big(\h{g}-q\h{f}\big)+\h{f}s=0,
\end{gather}
\end{subequations}
which admit Casoratian solutions given by \eqref{fg} and \eqref{hs} with entries \eqref{psi}.
This provides an alternative bilinear form for the coupled system \eqref{BT-1},
therefore both \eqref{BT-1-bil} and \eqref{BT-1-bil-2} can serve as bilinear forms for
{\rm H2} equation \eqref{H2} and the coupled system \eqref{uw-eq} or \eqref{uw-eqeq}.
\end{Corollary}

In addition, since $\varrho^{\pm}_i$ in \eqref{psi} are arbitrary parameters independent of $(n,m,l)$,
we introduce a dummy variable $x$ and replace $\varrho^{\pm}_i$ by ${\rm e}^{\pm k_i x}\varrho^{\pm}_i$, and then we have
\begin{gather}\label{psi-x}
\psi_{i}(x,l)=\varrho^{+}_i {\rm e}^{ k_i x} k_{i}^{l}(p+k_{i})^{n}(q+k_{i})^{m}
+\varrho^{-}_i {\rm e}^{- k_i x} (-k_{i})^{l}(p-k_{i})^{n}(q-k_{i})^{m},
\end{gather}
with which we redefine $\psi(x,l)=(\psi_1, \psi_2,\dots, \psi_N)^{\rm T}$.
Obviously,
\begin{gather*}
\psi(x,l+1)=\partial_x \psi(x,l).
\end{gather*}
Thus we have
\[
f=\big|\h{N-1}\big|=\big|\psi(x,l), \partial_x \psi(x,l), \partial_x^2 \psi(x,l),\dots, \partial_x^{N-1} \psi(x,l)\big|
=\big|\h{N-1}\big|_{[x]}^{},
\]
where the subscript $[x]$ means the determinant is a Wronskian composed by $\psi(x,l)$ and its derivatives
w.r.t.~$x$.
Besides, for the Casoratians $g$ and $\theta$ defined in \eqref{fg-theta}, we have
\[
g=f_x=\big|\h{N-2},N\big|_{[x]}^{},\qquad \theta =f_{xx}=\big|\widehat{N-3},N-1,N\big|_{[x]}^{}+\big|\widehat{N-2},N+1\big|_{[x]}^{}.
\]
Thus we reach the following.
\begin{Corollary}
In terms of Wronskians composed by $\psi(x,l)$ with entries \eqref{psi-x},
the bilinear system \eqref{BT-1-bil} is rewritten as
\begin{subequations}\label{BT-1-bil-3}
\begin{align}
& \big(D^2_x+2pD_x\big)f\cdot \t f=0,\label{BT-1-bil-3-a}\\
& \big(D^2_x+2qD_x\big)f\cdot \h f=0,
\end{align}
\end{subequations}
where $D$ is the Hirota bilinear operator defined as {\rm \cite{Hirota-1974}}
\[
D_x^i D_t^j f(x,t) \cdot g(x,t)=(\partial_x-\partial_{x'})^i (\partial_t-\partial_{t'})^j f(x,t) g(x',t')|_{x'=x,t'=t}.
\]
\end{Corollary}

\section{Continuous counterparts}\label{sec-5}

\subsection{Continuum limits}\label{sec-5-1}

Let us start to investigate continuum limits of the coupled system \eqref{uw-eqeq} via a two-step scheme,
which is based on the deformation of the H2 plane wave factor \cite{HZ-2009}:
\begin{align*}
\rho(n,m)=\biggl(\frac{p-k}{p+k}\biggr)^{n}\biggl(\frac{q-k}{q+k}\biggr)^{m}
& \longrightarrow {\rm e}^{-2k \xi +o(\frac{1}{p})}\biggl(\frac{q-k}{q+k}\biggr)^{m}
\\[.5ex]
& \longrightarrow {\rm e}^{-2k (\xi+\eta) -\frac{2 k^{3}}{3q^{2}}\eta + o\left(\frac{1}{q^3}\right)},
\end{align*}
where
\[\xi=\frac{n}{p},\qquad \eta=\frac{m}{q}\]
which are finite when $n\to \infty$, $p\to\infty$ and $m\to \infty$, $q\to\infty$.

In the first step, after Taylor expanding the coupled system \eqref{uw-eqeq}
into power series of $\frac{1}{p}$, the leading term in each equation w.r.t.~$\frac{1}{p}$ gives rise to
\begin{subequations}\label{uw-eq-sd}
\begin{gather}
\partial_{\xi}\big(w+\h w\big)+\big(w -\h w\big)\big(u -\h u +2 q\big)=0, \label{uw-eq-sd-a}
\\
\partial_{\xi}\big(u+ \h u\big) +\big(q+u-\h u\big)^{2}- q^{2}=0, \label{uw-eq-sd-b}
\end{gather}
\end{subequations}
where for convenience we still employ $u$ and $w$ to denote functions of $(\xi, m)$ in this step.
Note that equation~\eqref{uw-eq-sd-a} yields a form $u -\h u=A(w)$,
and substituting it into equation~\eqref{uw-eq-sd-b} yields a~form $\partial_{\xi}(u+ \h u)=B(w)$,
where $A(w)$ and $B(w)$ are some functions of $w$, $\h w$ and their derivatives w.r.t.~$\xi$.
Then, eliminating $u$ from them we get a differential-difference equation for~$w$, which is
\begin{gather}
\frac{1}{\big(w -\h w\big)^{2}}\Big[ 2q\big(w -\h w\big)\partial_{\xi}\big(w+\h w\big)+\partial_{\xi}\big(w^{2} -\h w^{2}\big)
+\big(\partial_{\xi}\big(w+\h w\big)\big)^{2}-\big(w -\h w\big)\partial_{\xi\xi}\big(w+\h w\big) \Big]\nonumber
\\ \qquad
{}+\frac{1}{\big(\h w -\hh w\big)^{2}}\Big[{-}2q\big(\h w -\hh w\big)\partial_{\xi}\big(\h w +\hh w\big)
+\partial_{\xi}\big(\h w^{2} -\hh w{}^{2}\big)-\big(\partial_{\xi}\big(\h w+\hh w\big)\big)^{2}\nonumber
\\ \qquad\hphantom{+\frac{1}{\big(\h w -\hh w\big)^{2}}\Big[}
{}-\big(\h w -\hh w\big) \partial_{\xi\xi}\big(\h w+\hh w\big)\Big]=0.
\label{w-eq-sd}
\end{gather}

In the second step, first, expanding \eqref{uw-eq-sd} in terms of $\frac{1}{q}$,
then, introducing
\begin{gather*}
x=\xi+\eta,\qquad t=\frac{\eta}{12q^2},
\end{gather*}
and reorganizing the expansion w.r.t.~$(x,t)$ and the derivatives w.r.t.\ them,
finally, from the leading terms we get
\begin{subequations}\label{uw-eq-cc}
\begin{align}
& w_{t}-6w_{x}u_{x}-w_{xxx}=0, \label{uw-eq-cc-a}\\
& u_{t}-6u_{x}^{2}-u_{xxx}=0, \label{uw-eq-cc-b}
\end{align}
\end{subequations}
which contains the potential KdV equation \eqref{uw-eq-cc-b} and its companion \eqref{uw-eq-cc-a}.
Again, one may eliminate $u$ from the system and obtain an equation of $w$:
\begin{gather}
w_{t}[6w_{2x}^{3}-10w_{x}w_{2x}w_{3x}+3w_{x}^{2}w_{4x}]+2w_{t}^{2}w_{x}w_{2x}
-3w_{xt}w_{x}[w_{x}w_{t}+2w_{2x}^{2}-2w_{x}w_{3x}]\nonumber
\\ \qquad
{}+3w_{x}^{2}w_{2x}w_{2xt}-2w_{x}^{3}w_{3xt}+w_{2t}w_{x}^{3}
+w_{2x}w_{3x}[8w_{x}w_{3x}-6w_{2x}^{2}]\nonumber
\\ \qquad
{}+6w_{x}w_{4x}[w_{2x}^{2}-w_{x}w_{3x}]-3w_{x}^{2}w_{2x}w_{5x}+w_{x}^{3}w_{6x}=0,
\label{w-eq-cc}
 \end{gather}
where the shorthand $w_{kx}$ means $\partial^k_x w$.

Later we will have a closer look at this system in Section~\ref{sec-5-3}.
We also note that higher order equations in the continuum limit
of the coupled quad-system \eqref{uw-eqeq}
will be explored in Appendix.

\subsection[Cauchy matrix approach to (5.3)]
{Cauchy matrix approach to (\ref{uw-eq-cc})}

A continuous version of the Cauchy matrix approach to the KdV equation has been developed in \cite{XZZ-2014}.
In the following we derive the coupled system \eqref{uw-eq-cc} using this approach.

\subsubsection{Derivation}

We will start from the Sylvester equation \eqref{Sylvester} and the function $S^{(i,j)}$
defined as in \eqref{S-ij}, while in this case $\bM$, $\br$ and $\bc^{\rm T}$ depends on $(x,t)$.
An auxiliary vector $\bu^{(i)}$ is introduced as
\begin{gather*}
\bu^{(i)} = (\bI+\bM)^{-1} \bK^i \br,\qquad i\in \mathbb{Z}.
 \end{gather*}

Let us first recall some useful relations obtained in \cite{XZZ-2014}. Proposition \ref{prop-2} holds as well in the continuous case.

{\samepage\begin{Proposition}\label{prop-4}\quad
\begin{enumerate}\itemsep=0pt
\item[$1.$]~$S^{(i,j)}$ satisfies the recursive relation for $i,j\in \mathbb{Z}$,
\begin{gather*}
S^{(i,j+2k)}=S^{(i+2k,j)}-\sum_{l=0}^{2k-1}(-1)^{l}S^{(2k-1-l,j)}S^{(i,l)},\qquad k=1,2,\dots,
\end{gather*}
and in particular, when $k=1$ we have
\begin{gather}\label{S-ij-k=1}
S^{(i,j+2)}=S^{(i+2,j)}-S^{(i,0)}S^{(1,j)}+S^{(i,1)}S^{(0,j)}.
\end{gather}
\item[$2.$] $\bK$ and any matrix similar to it lead to same $S^{(i,j)}$.
\item[$3.$] The symmetric property holds:
\begin{gather*}
 S^{(i,j)} = S^{(j,i)}.
\end{gather*}
\end{enumerate}
\end{Proposition}

}\noindent
In addition, a relation for $\bu^{(i)}$ is \cite{XZZ-2014}
\begin{gather*}
(\bI+\bM)\bK^{2k}\bu^{(i)} = \bK^{2k+i} \br
 -\sum_{l=0}^{2k-1}(-1)^{l}\bK^{2k-1-l}\br\bc^{\rm T} \bK^{l} \bu^{(i)}.
\end{gather*}
Multiplied $(\bI+\bM)^{-1}$ from the left, one first has
\[
\bK^{2k}\bu^{(i)}=(\bI+\bM)^{-1}\bK^{2k+i} \br-
\sum_{l=0}^{2k-1}(-1)^{l}\bigl((\bI+\bM)^{-1}\bK^{2k-1-l}\br\bigr) \bc^{\rm T} \bK^{l} \bu^{(i)},
\]
and then, in light of the definitions of $\bu^{(i)}$ and $S^{(i,l)}$, it follows that
\begin{gather}\label{u-i-rec}
\bu^{(2k+i)}=\bK^{2k}\bu^{(i)}+\sum_{l=0}^{2k-1}(-1)^{l}\bu^{(2k-1-l)}S^{(i,l)}.
\end{gather}

Introducing dispersion relation (cf.~\cite{XZZ-2014})
\begin{alignat*}{3}
& \br_{x}=\bK \br,\qquad&& \bc_{x}=\bK^{\rm T}\bc,&\\
& \br_{t}=4\bK^{3}\br,\qquad&& \bc_{t}=4\big(\bK^{\rm T}\big)^{3}\bc.&
\end{alignat*}
one has (cf.~\cite{XZZ-2014})
\begin{subequations}\label{u-i-xt}
\begin{gather}\label{u-i-x}
\bu_{x}^{(i)}=\bu^{(i+1)}-S^{(i,0)}\bu^{(0)},\\
\bu_{t}^{(i)}=4\big(\bu^{(i+3)}-S^{(i,0)}\bu^{(2)}-S^{(i,2)}\bu^{(0)}+S^{(i,1)}\bu^{(1)}\big),
\label{u-i-t}
\\
S_{x}^{(i,j)}=S^{(i+1,j)}+S^{(i,j+1)}-S^{(i,0)}S^{(0,j)},\nonumber 
\\
S_{t}^{(i,j)}=4\big(S^{(i+3,j)}+S^{(i,j+3)}+S^{(i,1)}S^{(1,j)}-S^{(i,0)}S^{(2,j)}-S^{(i,2)}S^{(0,j)}\big), \nonumber 
\end{gather}
\end{subequations}
and
\begin{gather*}
S_{xx}^{(i,j)}=S^{(i+2,j)}+S^{(i,j+2)}-2S^{(i+1,0)}S^{(0,j)}-2S^{(i,0)}S^{(0,j+1)}+2S^{(i+1,j+1)}
\\ \hphantom{S_{xx}^{(i,j)}=}
{}-S^{(i,0)}S^{(1,j)}-S^{(i,1)}S^{(0,j)}+2S^{(i,0)}S^{(0,0)}S^{(0,j)},
\\
S_{xxx}^{(i,j)}=S^{(i+3,j)}+S^{(i,j+3)}+3S^{(i+2,j+1)}+3S^{(i+1,j+2)}-3S^{(i+2,0)}S^{(0,j)}
\\ \hphantom{S_{xxx}^{(i,j)}=}
{}-3S^{(i,0)}S^{(0,j+2)}-6S^{(i+1,0)}S^{(0,j+1)}-3S^{(i+1,0)}S^{(1,j)}-3S^{(i,1)}S^{(0,j+1)}
\\ \hphantom{S_{xxx}^{(i,j)}=}
{}-S^{(i,2)}S^{(0,j)}-S^{(i,0)}S^{(2,j)}-3S^{(i+1,1)}S^{(0,j)}-3S^{(i,0)}S^{(1,j+1)}
\\ \hphantom{S_{xxx}^{(i,j)}=}
{}+6S^{(i+1,0)}S^{(0,0)}S^{(0,j)}+6S^{(i,0)}S^{(0,0)}S^{(0,j+1)}-2S^{(i,1)}S^{(1,j)}
+3S^{(i,0)}S^{(0,0)}S^{(1,j)}
\\ \hphantom{S_{xxx}^{(i,j)}=}
{}+6S^{(i,0)}S^{(1,0)}S^{(0,j)}
+3S^{(i,1)}S^{(0,0)}S^{(0,j)}-6S^{(i,0)}S^{(0,0)^{2}}S^{(0,j)}.
\end{gather*}
From the above general formulae, for $u$ and $w$ defined as in \eqref{uw-S}, one has
\begin{subequations}\label{u-xt}
\begin{gather}
u_{x}=2S^{(1,0)}-u^{2},\label{u-xt-x}
\\
u_{t}=4\big(2S^{(3,0)}-2uS^{(2,0)}+S^{(1,0)^{2}}\big),
\\
u_{xxx}=2S^{(3,0)}+6S^{(2,1)}-8uS^{(2,0)}-14S^{(1,0)^{2}}-6uS^{(1,1)} +24u^{2}S^{(1,0)}-6u^{4},
\end{gather}
\end{subequations}
and
\begin{subequations}\label{w-xt}
\begin{gather}
w_{x} = S^{(2,0)}+S^{(1,1)}-uw,\label{w-xt-x}\\
w_{t} = 4\big(S^{(4,0)}+S^{(1,3)}+wS^{(1,1)}-wS^{(2,0)}-uS^{(1,2)}\big), \\
w_{xx} = S^{(3,0)}+3S^{(1,2)}-2uS^{(2,0)}-uS^{(1,1)}-3w^{2}+2u^{2}w, \label{w-xt-xx} \\
w_{xxx}=S^{(4,0)}+4S^{(1,3)}+3S^{(2,2)}-3uS^{(3,0)}-13wS^{(2,0)}-4uS^{(2,1)}-8wS^{(1,1)}\nonumber
\\ \hphantom{w_{xxx}=}
{}+15uw^{2}+6u^{2}S^{(2,0)}+3u^{2}S^{(1,1)}-6u^{3}w.
\end{gather}
\end{subequations}
The function $u$ obeys the potential KdV equation (cf.~\cite{XZZ-2014})
\begin{gather}\label{pKdV}
u_{t}-6(u_{x})^{2}-u_{xxx}= 0,
\end{gather}
while $w$ satisfies
\begin{gather*}
w_{t}-6u_{x}w_{x}-w_{xxx}=3\big({-}S^{(2,2)}+S^{(4,0)}+uS^{(3,0)}+u^{2}S^{(1,1)}-wS^{(2,0)}-uw^{2}\big).
\end{gather*}
The right hand side vanishes by using the identity \eqref{S-ij-k=1} for $(i,j)=(2,0)$ and $(1,0)$.
Thus, we obtain equation~\eqref{uw-eq-cc-a}, i.e., \begin{gather} \label{w-eq-1}
w_{t}-6u_{x}w_{x}-w_{xxx}=0.
\end{gather}

Note that equation~\eqref{u-xt-x} indicates a Miura-type transformation
\begin{gather}\label{uw-MT}
2w=u_x+u^2,
\end{gather}
which connects the potential KdV equations \eqref{pKdV} and~\eqref{w-eq-1}.
For this we have the following.
\begin{Proposition}\label{prop-5}
In light of relation \eqref{uw-MT}, the two equations in the system \eqref{uw-eq-cc} are connec\-ted~as
 \begin{gather}\label{uw-eq-MT}
 w_{t}-6u_{x}w_{x}-w_{xxx} = \frac{1}{2}(\partial_{x}+2u) \big[u_{t}-6(u_{x})^{2}-u_{xxx}\big].
 \end{gather}
\end{Proposition}

\subsubsection{Lax Pair }

A Lax pair for the coupled system \eqref{uw-eq-cc}
can be constructed as well from the Cauchy matrix approach (cf.~\cite{HJN-2016}).
Let us consider the relation \eqref{u-i-xt} with $i=0$ and $i=1$, i.e., \begin{subequations}\label{bu-x}
\begin{gather}
\bu_{x}^{(0)}= \bu^{(1)}-u\bu^{(0)},
\\
\bu_{x}^{(1)}=\bu^{(2)}-w\bu^{(0)},
\end{gather}
\end{subequations}
and
\begin{subequations}\label{bu-t}
\begin{gather}
\bu_{t}^{(0)}=4\big(\bu^{(3)}-S^{(0,0)}\bu^{(2)}-S^{(0,2)}\bu^{(0)}+S^{(0,1)}\bu^{(1)}\big),
\\
\bu_{t}^{(1)}=4\big(\bu^{(4)}-S^{(1,0)}\bu^{(2)}-S^{(1,2)}\bu^{(0)}+S^{(1,1)}\bu^{(1)}\big).
\end{gather}
\end{subequations}
Then we do the following.
Denote the first entries of $\bu^{(0)}$ and $\bu^{(1)}$ by $\phi_1$ and $\phi_2$, respectively.
Since in the complex field matrix $\bK$ is always similar to a lower triangular matrix, in light of~Proposition \ref{prop-4},
without loss of generality, we assume in the following $\bK$ is lower triangular and denote the entry $K_{1,1}$ in $\bK$
by $\lambda$. Make use of the recursion relation \eqref{u-i-rec}
to express $\bu^{(2)}$, $\bu^{(3)}$ and $\bu^{(4)}$ in terms of $\bu^{(0)}$ and $\bu^{(1)}$.
Then, for $\Phi=(\phi_1,\phi_2)^{\rm T}$, from \eqref{bu-x} and \eqref{bu-t}, and after the above treatment
we arrive at
\begin{gather*}
\Phi_{x}= L_{1} \Phi,\qquad \Phi_{t}= L_{2} \Phi,
\end{gather*}
where
\begin{gather*}
L_{1}= \begin{pmatrix}
 -u & 1 \\
 \lambda^{2}-2w & u
 \end{pmatrix}\!,
 \\[1ex]
 L_{2}= 4 \begin{pmatrix}
 - \lambda^{2}u-w_{x} & \lambda^{2}+2w-u^{2}
 \\
 \lambda^{4}-2 \lambda^{2}w-\frac{1}{2}w_{xx}-uw_{x}& \lambda^{2}u+w_{x}
 \end{pmatrix}\!,
 \end{gather*}
and we have made use of \eqref{w-xt-x} and \eqref{w-xt-xx} to express
$S^{(1,1)}$, $S^{(0,2)}$, $S^{(1,2)}$ and $S^{(0,3)}$ in terms of~$u$,~$w$ and their derivatives.

The compatibility $\Phi_{xt}=\Phi_{tx}$ gives rise to
 \begin{gather*}
 L_{1_{t}}- L_{2_{x}}+[ L_{1}, L_{2}]=0,
 \end{gather*}
i.e., in explicit form, gives the following equations
\begin{gather*}
 4\lambda^{2}\big(u^2+u_x-2w \big) -u_{t}-8w \big(u^{2}-2w\big)-4uw_{x}+2w_{xx}=0, 
\\
 u^{2}-2w+u_{x}= 0, 
\\
 w_{t}-6u_{x}w_{x}-w_{xxx} +4w_{x}\big(u^2+u_x-2w \big) =0. 
\end{gather*}
This leads to three equations
\begin{subequations}\label{eq3}
\begin{gather}
2w=u^2+u_x,\label{eq3-1}\\
u_{t}-6(u_{x})^{2}-u_{xxx}=0, \label{eq3-2}\\
w_{t}-6u_{x}w_{x}-w_{xxx}=0, \label{eq3-3}
\end{gather}
\end{subequations}
where \eqref{eq3-1} is nothing but the Miura transformation \eqref{u-xt-x}, i.e., \eqref{uw-MT},
equation~\eqref{eq3-3} is a~consequence of \eqref{eq3-2} and \eqref{eq3-1} in light of Proposition \ref{prop-5}.

\subsection{Connection to the KdV}\label{sec-5-3}

Equation \eqref{uw-eq-cc-a} is the continuous counterpart of equation~\eqref{uw-eq-a}.
To reveal more links between this equation and the KdV equation, let us recall some known results about the KdV equation.
The potential KdV equation \eqref{uw-eq-cc-b} has a Lax pair\footnote{Note that equation~\eqref{uw-eq-cc-b} is a potential version of the usual KdV equation
$v_t=6vv_x+v_{xxx}$ with $v=2u_x$.}
\begin{subequations}\label{kdv-lax}
\begin{gather}
\phi_{xx}+2u_x\phi=\lambda \phi, \label{kdv-lax-a}
\\
 \phi_{t}=\phi_{xxx}+3(\lambda+2u_x)\phi_x.
\label{kdv-lax-b}
\end{gather}
\end{subequations}
Through the transformation
\begin{gather*}
u=(\ln f)_x
\end{gather*}
the potential KdV equation \eqref{uw-eq-cc-b} is bilinearized as \cite{Hirota-PRL--1971}
\begin{gather}\label{KdV-bil}
\big(D_xD_t-D_x^4\big)f\cdot f=0,
\end{gather}
from which one can construct a bilinear B\"acklund transformation (BT) \cite{Hirota-1974}
\begin{subequations}\label{kdv-BT-bil}
\begin{gather}
D^2_xf\cdot f'=\lambda f f',
\label{kdv-BT-bil-a}\\
\big(D_t -D_x^3-3 \lambda D_x\big)f\cdot f' =0,
\label{kdv-BT-bil-b}
\end{gather}
\end{subequations}
and if $f$ is a solution to the bilinear KdV equation \eqref{KdV-bil}, so is $f'$.
Introducing
\begin{gather*}
\phi=\frac{f'}{f},
\end{gather*}
one can recover the Lax pair \eqref{kdv-lax} from the BT \eqref{kdv-BT-bil} \cite{Hirota-1974}.

Note that by the transformation
\begin{gather*}
w=\frac{f'}{f},\qquad u=(\ln f)_x,
\end{gather*}
equation~\eqref{uw-eq-cc-a} is bilinearized as
\begin{gather*}
\big(D_t-D_x^3\big)f\cdot f'=0,
\end{gather*}
which belongs to a deformed or modified bilinear BT of the KdV equation \cite{ChenBC-2004}:
\begin{subequations}\label{kdv-BT-bil-D2}
\begin{gather}
\big(D_x^2+ 2\gamma D_x \big)f \cdot f'= 0,\label{kdv-BT-bil-D2-a}
\\
\big(D_t- D_x^3 \big)f \cdot f'= 0,\label{kdv-BT-bil-D2-b}
\end{gather}
\end{subequations}
where $\gamma$ is a new B\"acklund parameter.
Introducing
\begin{gather*}
\psi=\frac{f'}{f},\qquad u=(\ln f)_x,
\end{gather*}
from the above deformed bilinear BT, one finds
\begin{subequations}\label{kdv-lax-D}
\begin{align}
&\psi_{xx}-2\gamma \psi_x+2u_x\psi=0, \label{kdv-lax-D-a}
\\
& \psi_{t}=\psi_{xxx}+6u_x \psi_x,
\label{kdv-lax-D-b}
\end{align}
\end{subequations}
which is another Lax pair of the potential KdV equation \eqref{uw-eq-cc-b}.
The two Lax pairs, \eqref{kdv-lax} and~\eqref{kdv-lax-D}, are connected by the gauge transformation
\begin{gather*}
\phi={\rm e}^{-(\gamma x+4 \gamma^3 t)}\psi,\qquad \gamma^2=\lambda.
\end{gather*}

Now we get the links between equation~\eqref{uw-eq-cc-a} and the potential KdV equation \eqref{uw-eq-cc-b}:
the time-part \eqref{kdv-lax-D-b} of the Lax pair \eqref{kdv-lax-D} is nothing but the equation \eqref{uw-eq-cc-a},
which means $w$ is in some sense the eigenfunction of the (potential) KdV equation;
if, based on the idea in \cite{LeviB-PNAS-1980}, we~consider the transformation $f\to f'$
as a shift in a discrete independent variable, e.g., $f'=\t f$,
the transformation \eqref{kdv-BT-bil-D2-a} is converted to the bilinear equation \eqref{BT-1-bil-3-a}
with $p=\gamma$.

\section{Concluding remarks}\label{sec-6}

We have revisited H2 equation \eqref{H2} with the help of the Cauchy matrix approach and H1 equation.
The function $y$ in H2 can be expressed as~\eqref{yzuw}, where
$u=S^{(0,0)}$ and $w=S^{(1,0)}$ are variables defined in the Cauchy matrix approach
and $u$, $w$ satisfy the coupled system~\eqref{BT-1}.
This coupled system provides a BT between the $u$-equation, H1 \eqref{lpKdV}, and the 9-point $w$-equation~\eqref{w-eq}.
They compose a consistent triplet.
To investigate integrability of the 9-point equation~\eqref{w-eq} of~$w$,
we derived from~\eqref{BT-1} a coupled quad-system~\eqref{uw-eq}, or equivalently, equation~\eqref{uw-eqeq},
which is CAC and allows a Lax pair~\eqref{Lax-pair}.
This fact indicates the integrability of the 9-point equation~\eqref{w-eq}.
The coupled system \eqref{uw-eq} and the $\mathrm{H}1\times \mathrm{H}2$ system \eqref{H1H2}
derived in \cite{KNPT-2020} are connected through~\eqref{yzuw} and~\eqref{BT-2}.

H2 is a consequence of the coupled system \eqref{BT-1}.
Since in a consistent triplet the BT determines the other two equations in the triplet,
in light of Propositions~\ref{prop-333} and~\ref{prop-33}, one may consider the 9-point equation \eqref{w-eq}
is an alternative of H2 after removing the background $z^2_0$ and the term $z_0 u$ from $y$.
Again, due to the consistent triplets, the bilinear form \eqref{BT-1-bil} can serve as a bilinear form
for H2 \eqref{H2}, the coupled quad-system \eqref{uw-eq} and the 9-point equation \eqref{w-eq}.

As for the continuous counterparts, in continuum limits the coupled quad-system \eqref{uw-eqeq}
gives rise to the semi-discrete system \eqref{uw-eq-sd}, which implies equation~\eqref{w-eq-sd} for only $w$;
and the full continuum limit leads \eqref{uw-eq-sd} to \eqref{uw-eq-cc}.
The two equations in \eqref{uw-eq-cc} are connected via the Miura transformation \eqref{uw-MT} by \eqref{uw-eq-MT}.
Eliminating $u$ from \eqref{uw-eq-cc} yields equation~\eqref{w-eq-cc} for $w$ only.
In~addition to the Mirua-type connection with the potential KdV equation, equation~\eqref{uw-eq-cc-a} appears as the
$t$-part in the gauged Lax pair \eqref{kdv-lax-D} of the potential KdV equation \eqref{uw-eq-cc-b},
which indicates that $w$ can be considered as an eigenfunction (in the Lax pair) of the (potential) KdV equation.
Note that we derived the gauged Lax pair from a deformed bilinear BT \eqref{kdv-BT-bil-D2},
in~which~\eqref{kdv-BT-bil-D2-a} and the bilinear equation \eqref{BT-1-bil-3-a} share the same form
if we consider~$f'$ as a~shift~$\t f$.
We also note that higher order equations in the continuum limit
of the coupled quad-system~\eqref{uw-eqeq} will be investigated in Appendix~\ref{app}.

\appendix

\section{Higher order equations in the continuum limits}\label{app}

In \cite{WC-1987} Wiersma and Capel established a scheme to achieve higher order equations in continuum
limits, and they obtained the (potential) KdV hierarchy from the continuum limit of the lpKdV equation \eqref{lpKdV}.
In their scheme, discrete equations are handled in the (skew) coordinates
\[
(N=n+m, m),
\]
and in the first round, take
\begin{gather}\label{skew-1}
\delta=q-p\to 0,\qquad m\to \infty,\qquad \tau_1=\delta m \quad(\mathrm{being~finite}).
\end{gather}
For the plane wave factor $\rho$ given in Section~\ref{sec-5-1}, it turns out that
\begin{align*}
\rho &=\biggl(\frac{p-k}{p+k}\biggr)^{n}\biggl(\frac{q-k}{q+k}\biggr)^{m}
= \biggl(\frac{p-k}{p+k}\biggr)^{N}\biggl(1+\frac{2\delta k}{(p+\delta +k)(p-k)}\biggr)^{m}
 \\
 & \longrightarrow \biggl(\frac{p-k}{p+k}\biggr)^{N}\exp\biggl(\frac{2k\tau_1}{p^2-k^2}\biggr).
\end{align*}
The resulting semi-discrete equation usually has an evolution form
\begin{gather}\label{U-eq}
\partial_{\tau_1}V_N=f\big(V_N, {\rm e}^{j\partial_N}V_N\big),
\end{gather}
where we adopt the notation ${\rm e}^{j\partial_N}V_N=V_{N+j}$.
In the full continuum limit, infinitely many independent variables
$(t_1=x, t_3, t_5,\dots)$ are introduced through
\[
\rho=\biggl(\frac{p-k}{p+k}\biggr)^{N}\exp\biggl(\frac{2k\tau_1}{p^2-k^2}\biggr)
=\exp\Biggl({\sum^{\infty}_{l=0}k^{2l+1}t_{2l+1}}\Biggr),
\]
where
\begin{gather*}
t_{2l+1}=\frac{2}{p^{2l+1}}\bigg(\frac{-N}{2l+1}+\frac{\tau_1}{p}\bigg).
\end{gather*}
It then follows that
\begin{subequations}\label{t-pN}
\begin{gather}
\partial_N=\sum^{\infty}_{l=1}\frac{-2}{(2l+1)p^{2l+1}}\partial_{t_{2l+1}}
\end{gather}
and
\begin{gather}
\partial_{\tau_1}=\sum^{\infty}_{l=1}\frac{2}{p^{2l+2}}\partial_{t_{2l+1}}=\partial_p\partial_N.
\end{gather}
\end{subequations}
Replacing $\partial_{\tau_1}$ and $\partial_N$ in the equation \eqref{U-eq} with the above
operators and expanding the equation in terms of $1/p$,
a hierarchy of equations will be obtained from the coefficients of $1/p^j$ in the expansion.

With this scheme, the lpKdV equation \eqref{uw-eqeq-b} is rewritten in terms of $(N,m)$ as
\begin{gather*}
\big(p+q+V-\th{\t V}\big)\big(p-q+\th V-\t V\big)=p^2-q^2, 
\end{gather*}
where $V=V(N,m)=u(n,m)$.
In light of \eqref{skew-1}, it gives rise to a semi-discrete pKdV equation
\begin{gather}\label{sdkdv}
\partial_{\tau_1}V_N=1-\frac{2p}{2p+{\rm e}^{-\partial_N}_{}V_N-{\rm e}^{\partial_N}_{}V_N}.
\end{gather}
In the full limit, redefining $V=V(t_1=x, t_3, t_5,\dots)$ and using \eqref{t-pN}, the pKdV hierarchy is obtained by expanding \eqref{sdkdv}
in terms of $1/p$. The first two equations from the expansion are
\begin{gather}\label{KdV-U}
V_{t_3}+3(V_x)^2-V_{xxx}=0
\end{gather}
and
\begin{gather}\label{KdV5-U}
V_{t_5}-10(V_x)^3+5(V_{xx})^2+10 V_xV_{xxx}-V_{xxxxx}=0,
\end{gather}
which are the 3rd-order potential KdV (pKdV) equation and 5th-order pKdV equation.
For more details one may refer to \cite{WC-1987}.

For the equation \eqref{uw-eqeq-a}, in the coordinates $(N,m)$ it is rewritten as
\begin{gather*}
 \big(p-q+\th V-\t V\big)\big(W-\th{\t W}\big)-(p+q)\big(\t{W}-\th{W}\big)=0,
\end{gather*}
where $W=W(N,m)=w(n,m)$.
The semi-discrete equation obtained in light of \eqref{skew-1} is
\begin{gather*}
\partial_{\tau_1}W_N=\frac{{\rm e}^{-\partial_N}_{}W_N-{\rm e}^{\partial_N}_{}W_N}
{2p+{\rm e}^{-\partial_N}_{}V_N-{\rm e}^{\partial_N}_{}V_N},
\end{gather*}
where we have replaced $\partial_{\tau_1}V_N$ in the equation using \eqref{sdkdv}.
In the full limit of the above equation, the first two equations from the leading terms are
\begin{gather}\label{KdV-W}
W_{t_3}=W_{xxx}-3V_x W_x
\end{gather}
and
\begin{gather}\label{KdV5-W}
W_{t_5}=W_{xxxxx}-5V_xW_{xxx}-5V_{xx}W_{xx}-5V_{xxx}W_{x}+ 10(V_x)^2 W_x,
\end{gather}
where $W=W(t_1=x, t_3, t_5,\dots)$.

Noting that \eqref{KdV-U} and \eqref{eq3-2} are connected by $V=-2u$,
we introduce a gauge-transformed Schr\"odinger spectral problem (compared with \eqref{kdv-lax-D-a})
\begin{gather}
W_{xx}-2\gamma W_x- V_x W=0, \label{kdv-lax-U}
\end{gather}
where $W$ serves as the eigenfunction and $\gamma$ for the spectral parameter.
Then, it turns out that the pKdV equation \eqref{KdV-U} is a result of the compatibility
of \eqref{kdv-lax-U} and \eqref{KdV-W}, i.e., $(W_{xx})_{t_3}=(W_{t_3})_{xx}$,
and the compatibility of \eqref{kdv-lax-U} and \eqref{KdV5-W} yields the 5th-order pKdV equation \eqref{KdV5-U}.
The same correspondence holds for the 7th-order equations obtained from the full continuum limits.

\subsection*{Acknowledgments}

The authors are grateful to the referees for their invaluable comments.
This project is supported by the NSF of China (Nos.~11631007 and 11875040)
and Science and technology innovation plan of Shanghai (No.~20590742900).

\pdfbookmark[1]{References}{ref}
\LastPageEnding

\end{document}